\documentclass[conference]{IEEEtran}

\usepackage{verbatim}
\usepackage{epsfig,bbm}
\usepackage{CJK}
\usepackage{indentfirst}
\usepackage{multirow}
\usepackage{epstopdf}
\usepackage{graphicx}
\usepackage{footmisc}
\graphicspath{{./figures/}}
\usepackage{amsfonts}
\usepackage{mathrsfs}
\usepackage{setspace}
\usepackage{amsmath}
\usepackage{algorithm,algorithmic,amsbsy,amsmath,amssymb,epsfig,bbm,mathrsfs, bbm} 
\usepackage{amsthm}
\usepackage{verbatim}
\usepackage[subfigure]{graphfig}

\newtheorem{theorem}{Theorem}

\begin{document}

\title{Performance Analysis of Simultaneous Wireless Information and Power Transfer with Ambient RF Energy Harvesting}
 \author{Xiao Lu$^{\dagger}$,  Ian Flint$^{\ddagger}$, Dusit Niyato$^{\dagger}$, Nicolas Privault$^{\ddagger}$, and Ping Wang$^{\dagger}$	\\
 		~$^{\ddagger}$ School of Computer Engineering, Nanyang Technological University, Singapore \\
    ~$^{\dagger}$ School of Physical $\&$ Mathematical Sciences, Nanyang Technological University, Singapore  \\
   }

\markboth{}{Shell \MakeLowercase{\textit{et al.}}: Bare Demo of
IEEEtran.cls for Journals}

\maketitle
\begin{abstract}
 
The advance in RF energy transfer and harvesting technique over the past decade has enabled wireless energy replenishment for electronic devices, which is deemed as a promising alternative to address the energy bottleneck of conventional battery-powered devices. In this paper, by using a stochastic geometry approach, we aim to analyze the performance of an RF-powered wireless sensor in a downlink simultaneous wireless information and power transfer (SWIPT) system with ambient RF transmitters. Specifically, we consider the point-to-point downlink SWIPT transmission from an access point to a wireless sensor in a network, where ambient RF transmitters are distributed as a Ginibre $\alpha$-determinantal point process (DPP), which becomes the Poisson point process when $\alpha$ approaches zero. In the considered network, we focus on analyzing the performance of a sensor equipped with the power-splitting architecture. Under this architecture, we characterize the expected RF energy harvesting rate of the sensor. Moreover, we 
derive the upper bound of both power and transmission outage probabilities. Numerical results show that our upper bounds are accurate for different value of $\alpha$. 

\end{abstract}

\emph{Index terms- RF energy harvesting, SWIPT, power splitting, determinantal point process, Poisson point process, Ginibre model}  

\section{Introduction}
RF energy harvesting techniques have evolved as a promising and cost-effective solution to supply energy for wireless networks \cite{X.2014Lu,X.June2014Lu}. The research efforts over the past decade have advanced RF energy harvesting technique in circuit sensitivity, antenna efficiency, RF-to-DC conversion efficiency, and frequency range, etc \cite{XLuSurvey}.  
The recent development has also brought commercial products into the market. For example, the Powercaster transmitter and Powerharvester receiver \cite{Powercast} allow a transmission of 1W or 3W isotropic wireless power, and reception of the power by converting the harvested RF waves into electricity, respectively.
In this context, RF signals have been advocated to carry information as well as RF energy at the same time, which is referred to as the concept of simultaneous wireless information and power transfer (SWIPT) \cite{ZhangRuiMIMO}. Recently, SWIPT has drawn great research attention and been intensively investigated, for example, in SISO channel without and with co-channel interference, SISO relay channel, MISO broadcast system, MIMO broadcast system, and MIMO relay channel~\cite{XLuSurvey}. 


For performance analysis of large-scale RF energy harvesting networks, stochastic geometry is a suitable tool that characterizes random spatial patterns with point process. Poisson Point Process (PPP) modeling has been applied to analyze RF energy harvesting performance in cellular network~\cite{K2014Huang}, cognitive radio network~\cite{S2013Lee}, relay network~\cite{I2014Krikidis}, and network-coded cooperative network \cite{V.2014Mekikis}. The study in \cite{K2014Huang} investigates tradeoffs among transmit power and density of mobiles and wireless charging stations which are both modeled as a homogeneous PPP. The authors in~\cite{S2013Lee} study a cognitive radio network with energy harvesting secondary users, wherein both the primary and secondary networks are distributed as independent homogeneous PPPs. The maximum throughput of the secondary network has been characterized under the outage probability requirements for both primary and secondary networks. Reference \cite{I2014Krikidis} focuses on the impact of cooperative density and relay selection in a large-scale network with transmitter-receiver pairs distributed as a PPP.  In \cite{V.2014Mekikis}, the authors adopt PPP to model a two-way network-coded cooperative network with energy harvesting relays. The probability of successful data exchange and the network lifetime gain are derived in closed-form expressions.
Different from the above related work, our previous work in \cite{UL-SA-TS} adopts a more general analytical framework with Ginibre $\alpha$-determinantal point process (DPP) modeling, wherein the PPP is a special case when $\alpha$ approaches zero. Considering a stochastic network with ambient RF sources distributed following a Ginibre $\alpha$-DPP, we have investigated the uplink performance of an RF-powered sensor adopting separated receiver architecture, which equips the information receiver and RF energy harvester with independent antennas so that they function separately and observe different channel gains.
 

In this work, we continue to adopt the Ginibre $\alpha$-DPP modeling approach, which is suitable for modeling random phenomena where attraction/repulsion is observed. As attraction (or clustering) and repulsion are common behaviors in wireless communication systems, such as mobile cellular networks~\cite{S2013Cho} and mobile social networks~\cite{N2014Vastardis}, we aim to analyze network performance by characterizing different degrees of repulsion with the Ginibre $\alpha$-DPP.
In particular, we focus on the downlink performance of a point-to-point SWIPT system, where the receiver, i.e., an RF-powered sensor, performs information decoding and energy harvesting simultaneously. The considered sensor adopts the power-splitting architecture \cite{ZhangRuiMIMO}, which allows the information receiver and RF energy harvester to share the same antenna.
The sensor is assumed to be battery-free and operates based on the instant RF energy harvested from ambient RF transmitters.  Based on the considered model, we first characterize the expected RF energy harvesting rate (in Watt), then derive the upper bounds of both power and transmission outage probabilities in closed forms. The performance analysis provides a useful insight into the tradeoff among various network parameters.


{\em Notations}: Throughout the paper, we use $\mathbb{E}[X]$ to denote the probabilistic expectation of a random variable $X$, and $\mathbb{P}(A)$ to denote the probability of an event $A$. 

\subsection{Network Model}

We consider a battery-free sensor node harvesting energy from an access point and ambient RF transmitters. The power supply of the sensor solely comes from the instant harvested RF energy. Figure \ref{fig:systemmodel} shows the considered network model, where the sensor node harvests RF energy and utilizes the instantly harvested energy to power the circuit of the sensor.  We assume that the ambient RF transmitters, e.g., wireless routers and cellular mobiles, which can be deemed as RF energy sources for the sensor, are distributed as a general class of point processes, which will be specified in detail in Section~\ref{sec:geometricmodeling}.

The sensor is considered to adopt the power-splitting architecture \cite{ZhangRuiMIMO}, which enables the sensor to perform data transmission and RF energy harvesting simultaneously.  
As shown in Fig. \ref{fig:power-splitting}, with the power-splitting architecture, the sensor is equipped with a single antenna. By adopting a power splitter, this architecture splits the received RF signals into two streams for the information receiver and RF energy harvester respectively.  
After the power splitting, the portion of RF signals split to the energy harvester is denoted by $\eta$ (0 $\leq \eta \leq 1$), and that to the information receiver is $1-\eta$.

\begin{figure}
\centering
\includegraphics[width=0.35\textwidth]{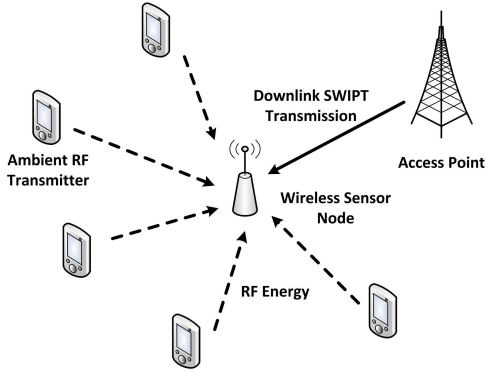}
\caption{A network model of ambient RF energy harvesting.} \label{fig:systemmodel}
\end{figure}

The RF energy harvesting rate of the sensor node from the access point in a free-space channel $P^{A}_{\mathrm{H}}$ can be obtained based on the Friis equation~\cite{Visser2013} \footnote{Other RF signal propagation models can also be used without loss of generality in the analysis of this paper.}
 as follows:
\begin{equation}
\label{eq:harvestedRFpower}
	P^{A}_{\mathrm{H}}	= \eta \beta P_{\mathrm{A}} \frac{G_{\mathrm{A}} G_{\mathrm{H}} \lambda^{2}_{A}}{(4\pi d_{A})^{2}},
\end{equation}
where $\beta$ is the RF-to-DC power conversion efficiency of the sensor node. 
$P_{\mathrm{A}}$, $G_{\mathrm{A}}$ and $\lambda_A$ are the transmit power, transmit antenna gain and transmitted  wavelength of the access point, respectively. $d_{A}$ is the distance between the  transmit antenna of the access point and the receiver antenna of the sensor node. $G_{\mathrm{H}}$ is the receive antenna gain of the sensor node. Let ${\mathbf{x}}_A \in {\mathbb{R}}^2$ be the coordinates of the access point $A$ in a referential centered at the sensor node. The distance can be obtained from $d_{A} = \epsilon + \lVert{\mathbf{x}}_A \rVert$, where $\epsilon$ is a fixed (small) parameter which ensures that the associated harvested RF power is finite in expectation. Physically, $\epsilon$ is the closest distance that the access point can locate near the sensor node.  


Let $P_{\mathrm{k}}$, $G_{\mathrm{k}}$ and $\lambda_k$ denote the transmit power, transmit antenna gain and transmitted  wavelength of the RF transmitter $k \in \mathcal{K}$, respectively. As the focus of this paper is to analyze the impact of the locations of ambient RF transmitters to the performance of the sensor node, similar to the related work \cite{G.2011Andrews}, we intentionally make some other parameters to be constants for ease of presentation and analysis.  Specifically, we have  $P_{\mathrm{k}}=P_S, G_{k}=G_{S}$, and $\lambda_k=\lambda$, for $k \in \mathcal{K}$.
Let ${\mathbf{x}}_k \in {\mathbb{R}}^2$ be the coordinates of the RF transmitter $k$.
Similar to \eqref{eq:harvestedRFpower}, we can calculate the RF energy harvesting rate from each RF transmitter $k \in \mathcal{K}$. Then, the aggregated RF energy harvesting rate by the sensor node can be computed as follows:
\begin{align}
\label{eq:PHpowersplitting}
P_{\mathrm{H}}^{PS}  =\eta \beta P_{\mathrm{A}} \frac{G_{\mathrm{A}} G_{\mathrm{H}} \lambda^{2}_{A}}{(4\pi  \lVert {\mathbf{x}}_A\rVert 	)^{2}}+	\sum_{ k \in {\mathcal{K}} }	\eta \beta P_{\mathrm{S}} \frac{G_{\mathrm{S}} G_{\mathrm{H}} \lambda^{2}}{(4\pi 	(\epsilon + \lVert {\mathbf{x}}_k \rVert)	)^{2}} 
\end{align}
where the second term represents the total energy harvesting rate from ambient RF transmitters. $\mathbf{x}_k$ denotes the location of RF transmitter $k$.

For the considered SWIPT system,  the scenarios of out-of-band transmission and in-band transmission need to be investigated. In the former, the access point transmits on a frequency band different from the one used for the RF energy harvesting (without co-channel interference). In the latter, the access point transmits on the same frequency band of ambient RF energy sources (with co-channel interference).
The downlink information rate at the sensor can be computed as in \eqref{eq:maxtransmissionPS}  \cite{XZhou2013},
\begin{figure*} 
\normalsize
\begin{eqnarray} 
\label{eq:maxtransmissionPS}
	C^{PS}=
	\begin{cases}
		  W \cdot \log_2	\left(	1 + h_{\mathrm A}	\frac { (1-\eta) P_{\mathrm{A}}	}	{ \xi I^{PS} +(1-\eta)\sigma^2+\sigma_{SP}^2}	\right) & \text{ if } P_{\mathrm H}^{PS} \geq P_{\mathrm C},\\
		0&\text{ if }P_{\mathrm H}^{PS} < P_{\mathrm C},
	\end{cases}
\end{eqnarray} 
\vspace{1pt} \hrulefill
\end{figure*} 
where $\xi$ in an indicator depending on whether we consider an out-of-band (i.e., $\xi=0$)
or in-band transmission scenario ($\xi=1$). $\sigma^2$ and $\sigma^2_{SP}$ represent the additive white Gaussian noise power and signal processing noise power, respectively. $I^{PS}$ denotes the interference from the ambient RF transmitters, which can be modeled as follows:
\begin{equation}
I^{PS} = \sum_{k \in \mathcal K}(1-\eta) P_{\mathrm{S}} \frac{G_{\mathrm{S}} G_{\mathrm{H}} \lambda^{2}}{(4\pi (\epsilon + \lVert {\mathbf{x}}_k \rVert)	)^{2}}.
\end{equation}

\begin{figure}
\centering
\includegraphics[width=0.322\textwidth]{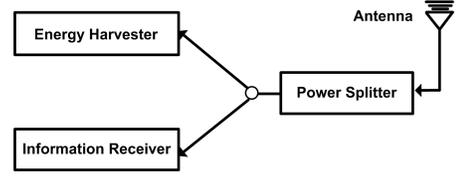}
\caption{Power-splitting receiver architecture.} \label{fig:power-splitting}
\end{figure}

\subsection{Stochastic Modeling of Ambient RF Transmitters} \label{sec:geometricmodeling}

We model the locations of RF transmitters using a point process $\mathcal{K}$ on an observation window $O:=\mathcal{B}(0,R)$ which is the closed ball centered at the origin and of radius $R>0$. In other terms, $\mathcal{K}$ is an almost surely finite random collection of points inside $\mathcal{B}(0,R)$. We refer to~\cite{Kallenberg} for the general theory of point processes. 

We focus on the Ginibre $\alpha$-DPP which is a type of $\alpha$-DPP (see~\cite{ShiraiTakahashi} for definitions and technical results). The Ginibre process is defined by the so-called Ginibre kernel given by
\begin{equation}
\label{eq:ginibre}
K(x,y)=\rho e^{\pi\rho x \bar{y}} e^{-\frac{\pi\rho}{2}( |x|^2 + |y|^2)},
\quad 
 x,y \in O=\mathcal{B}(0,R). 
\end{equation} 

We will write $\mathcal{K}\sim\mathrm{Det}(\alpha,K,\rho)$ 
when $\mathcal{K}$ is an $\alpha$-DPP with 
kernel $K$ defined in \eqref{eq:ginibre} and density with respect to the Lebesgue measure $\rho$. 
The spectral theorem for Hermitian and compact operators 
yields the following decomposition for the kernel of $K$:
\begin{equation*}
K(x,y)=\sum_{n\ge 0} \lambda_n \varphi_n(x)\overline{\varphi_n(y)},
\end{equation*}
where $(\varphi_i)_{i\ge 0}$ is a basis of $\mathrm{L}^2({O},\lambda)$, and $(\lambda_i)_{i\ge 0}$ the corresponding eigenvalues. 
 In e.g.~\cite{DecreusefondFlintVergne}, it is shown that
 the eigenvalues of the Ginibre point process on $O=\mathcal{B}(0,R)$ 
 are given by 
\begin{equation}
\label{eq:eigenvalues}
\lambda_n = \frac{\Gamma(n+1, \pi\rho R^2)}{n!},
\end{equation}
 where 
\begin{equation}
\label{eq:defgamma}
\Gamma(z,a) := \int_0^a e^{-t} t^{z-1}\,\mathrm{d}t, 
\qquad 
z \in \mathbb{C}, \quad a \ge 0, 
\end{equation}
 is the lower incomplete Gamma function. 
The eigenvectors of $K$ are given by
\begin{equation*}
\varphi_n(z) :=  \frac{1}{\sqrt{\lambda_n}}\frac{\sqrt{\rho}}{\sqrt{ n!}} e^{-\frac{\pi\rho}{2} | z |^2} (\sqrt{\pi\rho} z)^n, 
\qquad 
 n \in \mathbb{N}, 
\quad 
 z \in O. 
\end{equation*}
 We refer to~\cite{DecreusefondFlintVergne} for 
 further mathematical details on the Ginibre point process. 
 
Lastly, we emphasize that the Ginibre $\alpha$-DPP is stationary, in the sense that its distribution is invariant with respect to translations, c.f.~\cite{DecreusefondFlintVergne}. Hence, our choice of $O=\mathcal{B}(0,R)$ centered at the origin instead of $\mathbf{x}_i$ is justified.

\subsection{Performance Metrics} \label{sec:metrics}
 
We define the performance metrics of the sensor node as the expected of RF energy harvesting rate, power outage probability and transmission outage probability.  
The expected RF energy harvesting rate is defined as:
\begin{equation}\label{eq:expect}
 E_{P_H} \triangleq \mathbb{E} \left[ P^{PS}_{\mathrm{H}} \right]	.
\end{equation}

Power outage occurs when the sensor node becomes inactive due to lack of enough energy supply. The power outage probability is then defined as follows:
\begin{equation}
\label{eq:phi}
		P_{po} 	\triangleq	\mathbb{P} \left(	 P_{\mathrm{H}}	<	P_{\mathrm{C}}	\right)	,
\end{equation}
where $P_{\mathrm{C}}$ denotes the constant for power consumption of sensor node. Following practical models~\cite{G2009Miao}, the circuit power consumption of the sensor is assumed to be fixed.

Let $m \ge 0$ denote the minimum transmission rate requirement. If the sensor fails to meet this requirement, a transmission outage occurs. The transmission outage probability can be defined as follows:
\begin{equation}
\label{eq:psi}
	P_{to} \triangleq	\mathbb{P} \left(	C	<	m	\right)	.
\end{equation}

Specifically, for the above two outage probabilities, we focus on analyzing their upper bounds in this paper. 

\section{Performance Analysis}
\label{sec:Analysis}

In this section we estimate the metrics defined in Section~\ref{sec:metrics}
when $\mathcal{K}\sim\mathrm{Gin}(\alpha,\rho)$ is the Ginibre $\alpha$-DPP with parameter $\alpha=-1/j$, where $j\in\mathbb{N}^*$, and density $\rho>0$. 


\subsection{RF Energy Harvesting Rate}

The expected RF energy harvesting rate is evaluated as follows, which is similar to the result obtained in Theorem~1 in \cite{UL-SA-TS}. 
\begin{theorem}
\label{thm:evharvestedenergyPS}
The expected RF energy harvesting rate in the power-splitting architecture can be explicitly computed as 
\begin{align}
\label{eq:average_energy}
& \mathbb{E}[P_{\mathrm{H}}^{PS}] =\eta \beta P_A \frac{G_{\mathrm{A}} G_{\mathrm{H}} \lambda_A^{2}}{(4\pi 	 \lVert {\mathbf{x}}_A\rVert 	)^{2}}  \nonumber \\
& +2 \pi\eta \beta P_{\mathrm{S}} \frac{G_{\mathrm{S}}G_{\mathrm{H}} \lambda^{2}}{(4\pi )^{2}}
\rho\left(\frac{\epsilon}{R+\epsilon}+\ln(R+\epsilon)-1-\ln(\epsilon)\right)\\
& \approx_{\epsilon\rightarrow0}\frac{\rho\eta \beta P_{\mathrm{S}} G_{\mathrm{S}}G_{\mathrm{H}} \lambda^{2}}{8\pi}\ln\left(\frac R\epsilon\right). \label{eq:approximation_energy}
\end{align}
\end{theorem}
We recall some remarks which were made in \cite{UL-SA-TS}. First, we note that Theorem~\ref{thm:evharvestedenergyPS} implies that at the level of expectations, the Ginibre $\alpha$-DPP behaves like a homogeneous PPP and in particular, the expectation of RF energy harvesting rate is independent of the repulsion parameter $\alpha$. Therefore, on average, the harvested energy stays the same when $\alpha$ varies.  
\begin{proof}
The proof is the same as that of Theorem~1 in \cite{UL-SA-TS}, which we recall here for convenience. We have
\begin{align}
\mathbb{E}[P_{\mathrm{H}}^{PS}] & = \eta \beta P_A \frac{G_{\mathrm{A}}  G_{\mathrm{H}}  \lambda_A^{2}}{(4\pi 	 \lVert {\mathbf{x}}_A\rVert 	)^{2}} \nonumber \\ 
& + \eta \beta P_{\mathrm{S}} \frac{G_{\mathrm{S}}G_{\mathrm{H}} \lambda^{2}}{(4\pi )^{2}} \int_{O} \frac{\rho^{(1)}(x)}{(\epsilon+\|x\|)^2}\,\mathrm{d}x
\end{align}
by Campbell's formula~\cite{Kallenberg}, where 
$\rho^{(1)}(x)=K(x,x)=\rho$ 
 is the intensity function of $\mathcal{K}$ given by \cite{DecreusefondFlintVergne}.
We thus find
\begin{align}
\mathbb{E}[P_{\mathrm{H}}^{PS}] & =\eta \beta P_A \frac{G_{\mathrm{A}}  G_{\mathrm{H}}  \lambda_A^{2}}{(4\pi 	 \lVert {\mathbf{x}}_A\rVert 	)^{2}} \nonumber  \\
& +\eta \beta  P_{\mathrm{S}} \frac{G_{\mathrm{S}}G_{\mathrm{H}}  \lambda^{2}}{(4\pi )^{2}} 2 \pi \int_{0}^R \rho \frac{r}{(\epsilon+r)^2}\,\mathrm{d}r,
\end{align}
 by polar change of variable,
 and the integral on the r.h.s. is computed explicitly as follows:
\begin{equation*}
 \int_{0}^R\frac{r}{(\epsilon+r)^2}\,\mathrm{d}r=
\left(\frac{\epsilon}{R+\epsilon}+\ln(R+\epsilon)-1-\ln(\epsilon)\right),
\end{equation*}
which yields the result.
\end{proof}

\subsection{Power Outage Probability}

\begin{theorem}
\label{thm:upperboundoutage}
Let us define 
\begin{equation*}
\gamma^{PS}:=\frac\lambda{4\pi}\sqrt{ \frac{\eta \beta P_{\mathrm{S}}G_{\mathrm{S}} G_{\mathrm{H}} }{P_{\mathrm{C}} - P^A_{\mathrm H} } },
\end{equation*}
where $P^A_{\mathrm H}$ is defined in \eqref{eq:harvestedRFpower}.

If $P_{\mathrm C} \geq  P^A_{\mathrm H}$, then the following bound holds:
\begin{align} \label{power_outage_probability}
& \mathbb P \left(	 P_{\mathrm{H}}^{PS}	<	P_{\mathrm{C}}	\right) \nonumber \\
& \le  \left(\prod_{n\ge 0} \left(1+\alpha  \frac{\Gamma(n+1, \pi\rho\inf(R,\gamma^{PS})^2)}{n!}\right)\right)^{-1/\alpha},
\end{align}
where $\Gamma(z,a)$ is the lower incomplete Gamma function defined in \eqref{eq:defgamma}.

If $P_{\mathrm C} <  P^A_{\mathrm H}$, then  $\mathbb P \left(P_{\mathrm{H}}^{PS}	<	P_{\mathrm{C}}	\right)=0$.
\end{theorem}

\begin{proof}
Note that 
\begin{align} 
\label{eq:power_outage}
& \mathbb P\left(	 P_{\mathrm{H}}^{PS}	<	P_{\mathrm{C}}	\right) \nonumber \\
& =\mathbb P\left(	\sum_{ k \in {\mathcal{K}} }	\eta \beta P_{\mathrm{S}} \frac{G_{\mathrm{S}} G_{\mathrm{H}} \lambda^{2}}{(4\pi 	(\epsilon + \lVert {\mathbf{x}}_k \rVert)	)^{2}}<	P_{\mathrm{C}}- P^A_H	\right),
\end{align} 
whence it suffices to apply Theorem~2 of \cite{UL-SA-TS} to conclude.
\end{proof}




\subsection{Transmission Outage Probability}

In contrast with what was obtained in Theorem~\ref{thm:evharvestedenergyPS}, Theorem~\ref{thm:upperboundoutage} shows that the power outage probability depends on the repulsion parameter $\alpha$.

We begin by studying the in-band transmission scenario in the following theorem.
\begin{theorem}
\label{thm:boundpowersplittingPS}
Let us set
\begin{eqnarray} \label{eq:Tdef}
T =\frac{h_{\mathrm A}P_A}{2^{m/W}-1}-\sigma^2-\frac{\sigma_{SP}^2}{1-\eta }.
\end{eqnarray}
Assume that we are in the in-band scenario, i.e. $\xi=1$. The bound on the power outage probability is slightly different depending on the values of the parameters.
Specifically, if $P_{\mathrm C}-P^A_{\mathrm H}>0$,
then we obtain the bound in \eqref{eq:trans_outage}:
\begin{align} 
\label{eq:trans_outage}
& \mathbb P\left(	C^{PS}	<	m	\right) \nonumber  \\ 
& \le \left(\prod_{n\ge 0} \left(1+\alpha  \frac{\Gamma(n+1, \pi\rho\inf(R,\gamma^{PS})^2)}{n!}\right)\right)^{-1/\alpha} \nonumber  \\
&	+\frac{\rho P_{\mathrm{S}} G_{\mathrm{S}}G_{\mathrm{H}} \lambda^{2}\left(\frac{\epsilon}{R+\epsilon}+\ln(R+\epsilon)-1-\ln(\epsilon)\right)}{8\pi\max\left(T,\frac1{\eta \beta}\left( P_{\mathrm C}- P^A_H \right)\right)}.
\end{align} 

If  $P_{\mathrm C}- P^A_H\le0$, and $T\ge0$,
then
\begin{multline}
\label{eq:boundpsi2}
\mathbb P\left(	C^{PS}	<	m	\right)\\
\le\frac{\rho P_{\mathrm{S}} G_{\mathrm{S}}G_{\mathrm{H}} \lambda^{2}\left(\frac{\epsilon}{R+\epsilon}+\ln(R+\epsilon)-1-\ln(\epsilon)\right)}{8\pi T}.
\end{multline}

Lastly, if $\max\left(T, P_{\mathrm C}- P^A_H\right)\le0$, then we have
\begin{equation*}
	\mathbb P\left(	C^{PS}	<	m	\right)=1.
\end{equation*}
\end{theorem}

\begin{proof}
Using the definition of $C^{PS}$ given in \eqref{eq:maxtransmissionPS}, we find \eqref{eq:trans_outage_der}.
\begin{figure*} 
\normalsize
\begin{align} 
\label{eq:trans_outage_der}
& \mathbb P\left(	C^{PS}	<	m	\right) =\mathbb P\left(P_{\mathrm H}^{PS} < P_{\mathrm C}	\right) +\mathbb P\left(	C^{PS}	<	m,P_{\mathrm H}^{PS} \ge P_{\mathrm C}	\right) \nonumber \\
&	=\mathbb P\left(P_{\mathrm H}^{PS} < P_{\mathrm C}	\right) +\mathbb P\left( h_{\mathrm A}(1-\eta)P_A<\left((1-\eta)\sigma^2+\sigma_{SP}^2+I^{PS}\right)\left(2^{m/W}-1\right),P_{\mathrm H}^{PS} \ge P_{\mathrm C} \right)\nonumber\\
&	=\mathbb P\left(P_{\mathrm H}^{PS} < P_{\mathrm C}	\right)  
	\qquad+\mathbb P\left(\max\left(T, \frac1{\eta \beta}\left( P_{\mathrm C}- P^A_{\mathrm H}\right)\right)<\sum_{k \in \mathcal K} P_{\mathrm{S}} \frac{G_{\mathrm{S}} G_{\mathrm{H}} \lambda^{2}}{(4\pi (\epsilon + \lVert {\mathbf{x}}_k \rVert)	)^{2}}\right).
\end{align} 
\vspace{1pt} \hrulefill
\end{figure*}
By proceeding along the same lines as in the proof of Theorem~2 of \cite{UL-SA-TS}, we set
\begin{equation*}
f(\mathbf{x}_k):= P_{\mathrm{S}} \frac{G_{\mathrm{S}} G_{\mathrm{H}} \lambda^{2}}{(4\pi (\epsilon+\|\mathbf{x}_k\|))^{2}},
\end{equation*}
for $k\in\mathcal{K}$. Then,
\begin{align} 
\label{eq:transoutage_proof}
& \mathbb P\left(	C^{PS}	<	m	\right) 
	=\mathbb P\left(P_{\mathrm H}^{PS} < P_{\mathrm C}	\right)  \nonumber \\
&	+\mathbb P\left(\sum_{k\in\mathcal{K}} f(\mathbf{x}_k)>\max\left(T, \frac1{\eta \beta}\left( P_{\mathrm C}- P^A_{\mathrm H}\right)\right)\right)
\end{align} 
where $P^A_{\mathrm H}$ and $T$ are defined in \eqref{eq:harvestedRFpower} and \eqref{eq:Tdef}, respectively.
From this, we conclude that if  $\max\left( T,\frac1{\eta \beta}\left( P_{\mathrm C}-P^A_{\mathrm H}\right)\right)\le0$,
then
\begin{equation*}
	\mathbb P\left(	C^{PS}	<	m	\right)=1.
\end{equation*}
Otherwise, if $\max\left(T,\frac1{\eta \beta}\left( P_{\mathrm C}-P^A_{\mathrm H}\right)\right)>0$, then by Markov's inequality, we have
\begin{align}
&	\mathbb P\left(	C^{PS}	<	m	\right) \le \mathbb P\left(P_{\mathrm H}^{PS} < P_{\mathrm C}	\right)  \nonumber \\
&	+\frac{1}{\max\left( T,\frac1{\eta \beta}\left( P_{\mathrm C}- P^A_{\mathrm H} \right)\right)}\mathbb E\left[\sum_{k\in\mathcal{K}} f(\mathbf{x}_k)\right],
\end{align}
which can be computed as
\begin{align}
\label{eq:intermequationofthm3}
	& \mathbb P\left(	C^{PS}	<	m	\right)  \le\mathbb P\left(P_{\mathrm H}^{PS} < P_{\mathrm C}	\right) \nonumber \\
	& +\frac{\rho P_{\mathrm{S}} G_{\mathrm{S}}G_{\mathrm{H}} \lambda^{2}\left(\frac{\epsilon}{R+\epsilon}+\ln(R+\epsilon)-1-\ln(\epsilon)\right)}{8\pi\max\left( T,\frac1{\eta \beta}\left( P_{\mathrm C}-P^A_{\mathrm H} \right)\right)},
\end{align}
by Theorem~1 of \cite{UL-SA-TS}. As for the first term of \eqref{eq:intermequationofthm3}, it is upper-bounded by a straightforward application of Theorem~\ref{thm:upperboundoutage}, and we find \eqref{eq:intermequation2ofthm3},
\begin{figure*} 
\normalsize
\begin{align} 
\label{eq:intermequation2ofthm3}
 \mathbb P\left(	C^{PS}	<	m	\right)  
 \le\left(\prod_{n\ge 0} \left(1+\alpha  \frac{\Gamma(n+1, \pi\rho\inf(R,\gamma^{PS})^2)}{n!}\right)\right)^{-1/\alpha}1_{\left\{P_{\mathrm C}-P^A_{\mathrm H} \ge0\right\}}   
 +\frac{\rho P_{\mathrm{S}} G_{\mathrm{S}}G_{\mathrm{H}} \lambda^{2}\left(\frac{\epsilon}{R+\epsilon}+\ln(R+\epsilon)-1-\ln(\epsilon)\right)}{8\pi\max\left( T,\frac1{\eta \beta}\left( P_{\mathrm C}- P^A_{\mathrm H}\right)\right)},
\end{align} 
\vspace{1pt} \hrulefill
\end{figure*}
where $1_A$ is the indicator function of a set $A$, i.e. the functional equal to $1$ on $A$ and equal to $0$ elsewhere.

Note that the right hand side of \eqref{eq:intermequation2ofthm3} might be larger than $1$, and in that case we do better than the trivial inequality $\mathbb P\left(	C^{PS}	<	m	\right)\le 1$. Note also that \eqref{eq:intermequation2ofthm3} is a compact notation of the different cases discussed in Theorem~\ref{thm:boundpowersplittingPS}.
\end{proof}

For the out-of-band transmission scenario, an upper bound of the transmission outage probability can also be derived based on Markov inequality, following similar steps similar to those in Theorem~\ref{thm:boundpowersplittingPS}. Due to the space limit, we omit it in this paper.


\section{Numerical Results}


We assume that all the ambient RF transmitters are LTE-enabled mobiles operating on the typical $1800MHz$ frequency. The corresponding wavelength $\lambda$ is $0.167m$.  The circuit power consumption $P_{\mathrm C}$ is fixed to be $-18dBm$ (i.e., $15.8\mu W$) as in~\cite{UL-SA-TS}. The other parameters adopted in the simulations are shown in Table \ref{parameter_setting} unless specified otherwise.
Note that the results for the PPP case are identical to that of the $\alpha$-DPP, when $\alpha=0$. The numerical results presented in this section are averaged over $10^5$ simulation runs.

We interpret the upper-bounds derived in the previous section as worst-case scenarios.
This leads us to perform the simulations of this section in a different regime, in an attempt to approach the upper bounds. 
The simulation under this regime is known to perfectly approach the upper-bound of Theorem~\ref{thm:upperboundoutage}, whereas there is still a gap compared to the upper-bound of Theorem~\ref{thm:boundpowersplittingPS}.

\begin{table}
\centering
\caption{\footnotesize Parameter Setting.} \label{parameter_setting}
\begin{tabular}{|l|l|l|l|l|l|l|l|} 
\hline
Symbol & $G_{S}$,$G_{A}$,$G_{H}$ & $\beta$  & $P_{S}$,$P_{A}$ & $W$ &    $\sigma^2$,$\sigma^2_{SP}$ \\ 
\hline
Value  & 1.5 &  0.3 & 1W  & 10KHz  &  -90dBm \\
\hline              
\end{tabular}
\end{table}


Figure \ref{fig:amount_of_energy} shows the expected RF energy harvesting rate versus density of ambient RF transmitters. We can see that the simulation results, which were done in the general scenario described in this paper, match the analytical expression (\ref{eq:average_energy}) accurately over a wide range of transmitter densities $\rho$, i.e., from $0.01$ to $1$. 
We observe that when $\epsilon=0.001$, the sensor achieves larger RF energy harvesting rate than that in the case when $\epsilon=0.1$. 
This result is expected since, from \eqref{eq:PHpowersplitting}, the smaller the distance $\epsilon$ (i.e, the RF transmitters can be located near the sensor) the more aggregated RF energy harvesting rate is available. 
We also find that the difference between the exact analytical results obtained from (\ref{eq:average_energy}) and approximate results obtained from \eqref{eq:approximation_energy} is also dependent on $\epsilon$. Specifically, a smaller $\epsilon$ results in a more accurate approximation. As shown in Fig.~\ref{fig:amount_of_energy}, compared to when $\epsilon=0.1$, the approximate results more closely approach  the analytical results when $\epsilon=0.001$.  

\begin{figure}
\centering
\includegraphics[width=0.4\textwidth]{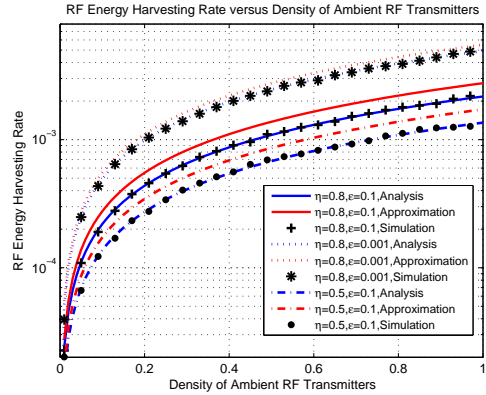}
\caption{RF energy harvesting rate versus density of ambient RF transmitters. } \label{fig:amount_of_energy}
\vspace{-5mm}
\end{figure}

Figure~\ref{fig:poweroutage_rho} shows the upper bound of the power outage probability versus the density of ambient RF transmitters, when $\eta=0.5$ and $\eta=1$. We observe that $P_{po}$ is a decreasing function of the transmitter density $\rho$. The numerical results, which were done in a worst-case scenario, are shown to approach the analytical expression in (\ref{power_outage_probability}) accurately for different settings of $\alpha$. Moreover, a larger repulsion among the location of the RF transmitters (i.e., smaller $\alpha$) results in a lower power outage probability. 

\begin{figure}
\centering
\includegraphics[width=0.4\textwidth]{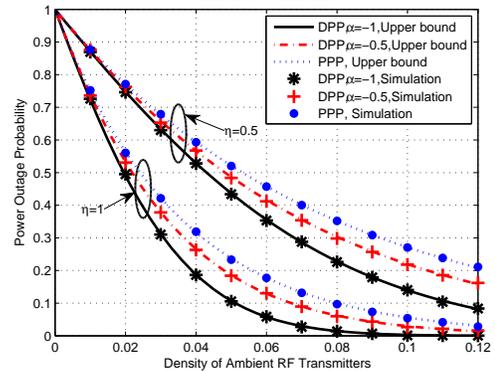}
\caption{Upper bound of power outage probability versus density of ambient RF transmitters.} \label{fig:poweroutage_rho}
\end{figure}

Figures \ref{fig:transoutage_rho} and \ref{fig:transoutage_d} present the upper bound of the transmission outage probability versus the density of RF transmitters and versus the distance between the access point and the sensor, respectively, in an in-band transmission scenario. We set the transmission rate requirement to be 0.02kbps and $\eta$ to be 0.5. 
In Fig. \ref{fig:transoutage_rho}, 
we see that, with the increase of the transmitter density, the upper bound of the transmission outage probability first decreases quickly, then begins to rebound slowly after a certain point. The rebound effect is caused by the increased interference due to the growth in transmitter density, which lowers the achievable transmission rate.       
Similar to Fig. \ref{fig:transoutage_rho}, we observe in Fig. \ref{fig:transoutage_d} that our analytical bound is tight when $d_A$ is small, and becomes more relaxed with the increase of $d_A$. It can be seen that when $d_A$ is small, there is no transmission outage. This is because when the access point locates within a certain close range near the sensor, the sensor can receive not only enough information rate but also sufficient power from the SWIPT transmission alone, regardless of the ambient transmitter density. Thus, the transmission outage probability equals zero in this case.

\begin{figure}
\centering
\includegraphics[width=0.4\textwidth]{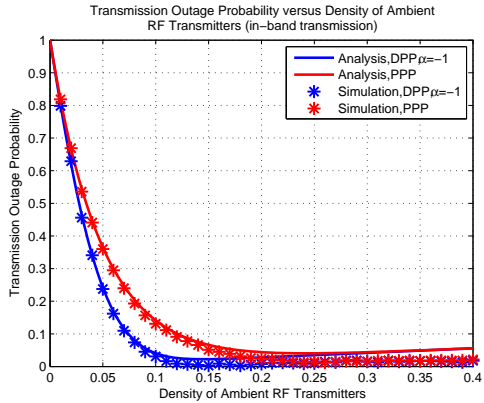}
\caption{Upper bound of transmission outage probability versus density of ambient RF transmitters.} \label{fig:transoutage_rho}
\vspace{-5mm}
\end{figure}

\begin{figure}
\centering
\includegraphics[width=0.4\textwidth]{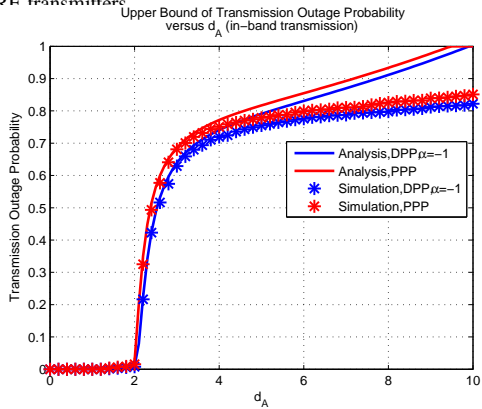}
\caption{Upper bound of transmission outage probability versus the distance $d_A$ ($\rho=0.02$).} \label{fig:transoutage_d}
\vspace{-5mm}
\end{figure}

\section{Conclusion}

We have analyzed the performance of an RF-powered sensor network in a downlink SWIPT system with ambient RF transmitters. We have adopted a repulsive point process, called a Ginibre $\alpha$-determintal point process, which allows to model a network where the locations of the RF transmitters demonstrate repulsion. We have derived the expression of the expected RF energy harvesting rate of the RF-power sensor. We have also characterized the worst-case performance of the sensor node in terms of the upper bounds of power and transmission outage probabilities. The performance evaluation shows that the exact analytical results and simulation results are well matched with the simulation results. Therefore, the proposed analysis will be useful in practice. Our future work will extend the performance analysis from a single-antenna network to a multi-antenna network. Another direction is to explore the network performance in other metrics, such as downlink coverage probability of the access point.


\section*{Acknowledgements}
This work was supported in part by Singapore MOE Tier 1 (RG18/13 and  RG33/12), MOE Tier 2 (MOE2013-T2-2-067) and MOE Tier 2 (ARC3/13).

\end{document}